\documentclass[a4paper,UKenglish,cleveref, autoref, thm-restate]{lipics-v2021}

\title{A Linear Time Algorithm for Constructing Hierarchical Overlap Graphs}

\author{Sangsoo Park}{Samsung Electronics, Korea}{cki86201@gmail.com}{https://orcid.org/0000-0002-6593-4336}{}

\author{Sung Gwan Park}{Samsung Electronics, Korea}{sgpark@theory.snu.ac.kr}{https://orcid.org/0000-0002-3255-9752}{}

\author{Bastien Cazaux}{LIRMM, Univ Montpellier, CNRS, France}{bastien.cazaux@lirmm.fr}{https://orcid.org/0000-0002-1761-4354}{}

\author{Kunsoo Park}{Seoul National University, Korea}{kpark@theory.snu.ac.kr}{https://orcid.org/0000-0001-5225-0907}{}

\author{Eric Rivals}{LIRMM, Univ Montpellier, CNRS, France}{rivals@lirmm.fr}{https://orcid.org/0000-0003-3791-3973}{}

\authorrunning{S. Park, S.\,G. Park, B. Cazaux, K. Park, and E. Rivals}
\Copyright{Sangsoo Park, Sung Gwan Park, Bastien Cazaux, Kunsoo Park, and Eric Rivals}



\category{} 

\relatedversion{} 

\nolinenumbers 

\EventEditors{John Q. Open and Joan R. Access}
\EventNoEds{2}
\EventLongTitle{42nd Conference on Very Important Topics (CVIT 2016)}
\EventShortTitle{CVIT 2016}
\EventAcronym{CVIT}
\EventYear{2016}
\EventDate{December 24--27, 2016}
\EventLocation{Little Whinging, United Kingdom}
\EventLogo{}
\SeriesVolume{42}
\ArticleNo{23}

\newcommand{\norm}[1]{\vert \vert #1 \vert \vert }

\newcommand{\hog}{\text{HOG}}
\newcommand{\ehog}{\text{EHOG}}

\usepackage{algorithm}
\usepackage[noend]{algpseudocode}

\begin{document}

\maketitle

\begin{abstract}
The hierarchical overlap graph (HOG) is a graph that encodes overlaps from a given set $P$ of $n$ strings, as the overlap graph does. A best known algorithm constructs HOG in $O(\norm{P} \log n)$ time and $O(\norm{P})$ space, where $\norm{P}$ is the sum of lengths of the strings in $P$. In this paper we present a new algorithm to construct HOG in $O(\norm{P})$ time and space. Hence, the construction time and space of HOG are better than those of the overlap graph, which are $O(\norm{P} + n^2)$.  
\end{abstract}

\section{Introduction}

For a set of strings, a \emph{superstring} of the set is a string that has all strings in the set as a substring. The \emph{shortest superstring} problem is to find a shortest superstring of a set of strings. This problem is known to play an important role in \emph{DNA assembly}, which is the problem to restore the entire genome from short sequencing reads. Despite its importance, the shortest superstring problem is known to be NP-hard \cite{Gallant80}. As a result, extensive research has been done to find good approximation algorithms for the shortest superstring problem \cite{Blum94,Sweedyk00,Mucha13,Paluch14,Tarhio88,Ukkonen90}.

The shortest superstring problem is reduced to finding a shortest hamiltonian path in a graph that encodes overlaps between the strings \cite{Blum94,sga_myers, Pevzner}, which is the \emph{distance graph} or equivalent \emph{overlap graph}. The overlap graph \cite{Peltola-83} of a set of strings is a graph in which each string constitutes a node and an edge connecting two nodes shows the longest overlap between them. Many approaches for approximating the shortest superstring problem focus on the overlap graph, and try to find good approximations of its hamiltonian path \cite{Mucha13,Paluch14}. 

Given a set of strings $P = \{s_1, s_2, ..., s_n\}$, computing the overlap graph of $P$ is equivalent to solving the \emph{all-pair suffix-prefix problem}, which is to find the longest overlap for every pair of strings in $P$. The best theoretical bound for this problem is $O(\norm{P} + n^2)$ \cite{Gusfield92}, where $\norm{P}$ is the sum of lengths of the strings in $P$. Since the input size of the problem is $O(\norm{P})$ and the output size is $O(n^2)$, this bound is optimal. There has also been extensive research on the all-pair suffix-prefix problem in the practical point of view \cite{Gonnella12,Lim17,Rachid15} because it is the first step in DNA assembly.

Recently, Cazaux and Rivals \cite{jda2016, hogipl} proposed a new graph which stores the overlap information, called the \emph{hierarchical overlap graph} (HOG). HOG is a graph with two types of edges (which will be defined in Section 2) in which a node represents either a string or the longest overlap between a pair of strings. The \emph{extended hierarchical overlap graph} (EHOG) is also a graph with two types of edges in which a node represents either a string or an overlap between a pair of strings (which may be not the longest one). For example, Figure \ref{fig:hog} shows EHOG and HOG built with $P = \{aacaa, aagt, gtc\}$. Even though HOG and EHOG may be the same for some input instances, there are the instances where the ratio of EHOG size over the HOG size tend to infinity with respect to the number of nodes. Therefore, HOG has an advantage over EHOG in both practical and theoretical points of view.

\begin{figure}[t]
    \centering
    \includegraphics[width=1\linewidth]{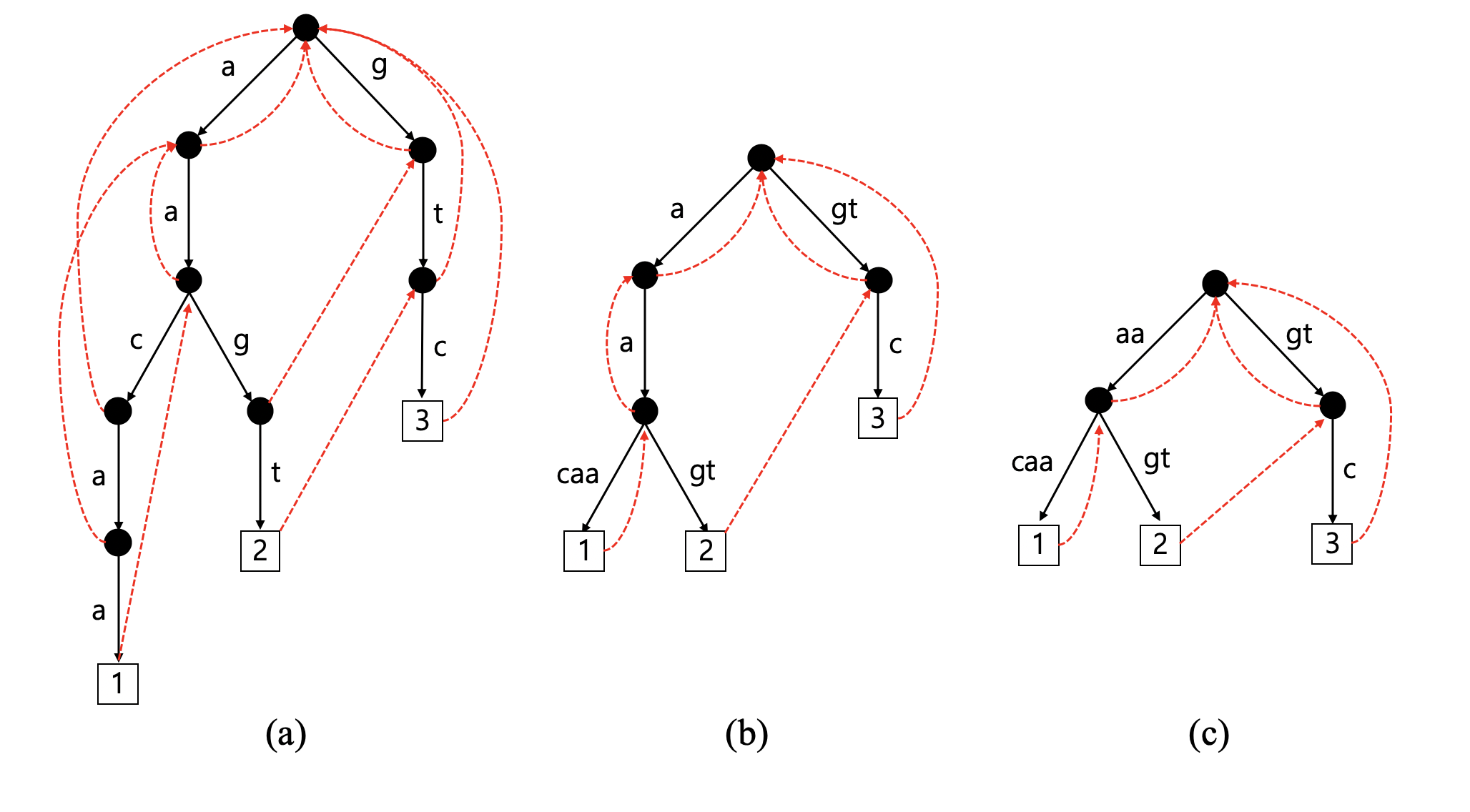}
    \caption{Data structures built with $P=\{aacaa, aagt, gtc\}$. (a) Aho-Corasick trie. (b) Extended Hierarchical Overlap Graph. (c) Hierarchical Overlap Graph.}
    \label{fig:hog}
\end{figure}

HOG also has an advantage compared to the overlap graph \cite{hogipl}. First, HOG uses only $O(\norm{P})$ space, while the overlap graph needs $O(\norm{P} + n^2)$ space in total. For input instances with many short strings, HOG uses a considerably smaller amount of space than the overlap graph. Second, HOG contains the relationship between the overlaps themselves, since the overlaps appear as nodes in HOG. In contrast, the overlap graph stores only the lengths of the longest overlaps, and thus we cannot find the relationship between two overlaps easily. Therefore, HOG stores more information than the overlap graph, while using less space.

There have been many works to compute HOG and EHOG efficiently. Computing the EHOG from $P$ costs $O(\norm{P})$ time, which is optimal \cite{cyccover}. For computing the HOG, Cazaux and Rivals proposed an $O(\norm{P} + n^2)$ time algorithm using $O(\norm{P}+n\times\min(n,\max\{|s|:s\in P\}))$ space \cite{hogipl}. Recently, Park et al.~\cite{sgpark} gave an $O(\norm{P} \log n)$ time algorithm using $O(\norm{P})$ space by using the segment tree data structure.

In this paper we present a new algorithm to compute HOG, which uses $O(\norm{P})$ time and space, which are both optimal. The algorithm is based on the Aho-Corasick trie \cite{Aho-Corasick} and the border array \cite{KMP}.
Therefore, the construction time and space of HOG are better than those of the overlap graph, which are $O(\norm{P}+n^2)$, and this fact may lead to many applications of HOG. 
For example, consider the problem of finding \emph{optimal cycle cover} in the overlap graph built with a set $P = \{s_1, s_2, ..., s_n\}$ of strings.
Typically this problem needs to be solved in finding good approximations of shortest superstrings. 
A greedy algorithm to solve the optimal cycle cover problem on the overlap graph was given in \cite{Blum94}, which takes $O(\norm{P} + n^2)$ time. Recently, Cazaux and Rivals proposed an $O(\norm{P})$ time algorithm to solve the optimal cycle cover problem given the HOG or EHOG of $P$ \cite{jda2016}. 
By using our result in this paper, the optimal cycle cover problem can be solved in $O(\norm{P})$ time and space by using HOG instead of the overlap graph.

The rest of the paper is organized as follows.
In Section 2 we give preliminary information for HOG and formalize the problem. In Section 3 we present an $O(\norm{P})$ time and space algorithm for computing HOG. In Section 4 we conclude and discuss a future work.

\section{Preliminaries}

\subsection{Basic notation}
We consider strings over a constant-size alphabet $\Sigma$. The length of a string $s$ is denoted by $|s|$. Given two integers $1 \le l \le r \le |s|$, the substring of $s$ which starts from $l$ and ends at $r$ is denoted by $s[l..r]$. Note that $s[l..r]$ is a prefix of $s$ when $l = 1$, and a suffix of $s$ when $r = |s|$. If a prefix (suffix) of $s$ is different from $s$, we call it a proper prefix (suffix) of $s$. Given two strings $s$ and $t$, an \emph{overlap} from $s$ to $t$ is a string which is both a proper suffix of $s$ and a proper prefix of $t$. Given a set $P = \{s_1, s_2, ..., s_n\}$ of strings, the sum of $|s_i|$'s is denoted by $||P||$.

\subsection{Hierarchical Overlap Graph}

We define \emph{hierarchical overlap graph} and \emph{extended hierarchical overlap graph} as in \cite{hogipl}.

\begin{definition}

\emph{(Hierarchical Overlap Graph) Given a set $P = \{s_1, s_2, \ldots, s_n\}$, we define $Ov(P)$ as the set of the \emph{longest} overlap from $s_i$ to $s_j$ for $1 \le i, j \le n$. The \emph{hierarchical overlap graph} of $P$, denoted by $\hog(P)$, is a directed graph with a vertex set $V = P \cup Ov(P) \cup \{\epsilon\}$ and an edge set $E = E_1 \cup E_2$, where $E_1 = \{(x, y) \in V \times V \mid x$ is the longest proper prefix of $y\}$ and $E_2 = \{(x, y) \in V \times V \mid y$ is the longest proper suffix of $x\}$.}
\end{definition}

\begin{definition}

\emph{(Extended Hierarchical Overlap Graph) Given a set $P = \{s_1, s_2, \ldots, s_n\}$, we define $Ov^+(P)$ as the set of all overlaps from $s_i$ to $s_j$ for $1 \le i, j \le n$. The \emph{extended hierarchical overlap graph} of $P$, denoted by $\ehog(P)$, is a directed graph with a vertex set $V^+ = P \cup Ov^+(P) \cup \{\epsilon\}$ and an edge set $E^+ = E^+_1 \cup E^+_2$, where $E^+_1 = \{(x, y) \in V^+ \times V^+ \mid x$ is the longest proper prefix of $y\}$ and $E^+_2 = \{(x, y) \in V^+ \times V^+ \mid y$ is the longest proper suffix of $x\}$.}
\end{definition}

Figure \ref{fig:hog} shows the Aho-Corasick trie \cite{Aho-Corasick}, EHOG, and HOG built with $P = \{aacaa, aagt, gtc\}$. It is shown in \cite{hogipl} that EHOG is a contracted form of the Aho-Corasick trie and HOG is a contracted form of EHOG. 

As in the Aho-Corasick trie, each node $u$ in HOG or EHOG corresponds to a string (denoted by $S(u)$), which is a concatenation of all labels on the path from the root (node representing $\epsilon$) to $u$.

There are two types of edges in EHOG and HOG as in the Aho-Corasick trie: a tree edge and a failure link. An edge $(u, v)$ is a tree edge (an edge in $E^+_1$ or $E_1$, solid line in Figure \ref{fig:hog}) in an EHOG (HOG), if $S(u)$ is the longest proper prefix of $S(v)$ in the EHOG (HOG). It is a failure link (an edge in $E^+_2$ or $E_2$, dotted line in Figure \ref{fig:hog}) in an EHOG (HOG), if $S(v)$ is the longest proper suffix of $S(u)$ in the EHOG (HOG).

Given a set $P = \{ s_1, s_2, \ldots, s_n\}$ of strings, we can build an EHOG of $P$ in $O(||P||)$ time and space \cite{hogipl}. Furthermore, given $\ehog(P)$ and $Ov(P)$, we can compute $\hog(P)$ in $O(||P||)$ time and space \cite{hogipl}. Therefore, the bottleneck of computing $\hog(P)$ is computing $Ov(P)$ efficiently.

\section{Computing HOG in linear time}

In this section we introduce an algorithm to build the HOG of $P=\{s_1, s_2, \ldots, s_n\}$ in $O(\norm{P})$ time. We assume that there are no two different strings $s_i, s_j \in P$ such that $s_i$ is a substring of $s_j$ for simplicity of presentation. Our algorithm directly computes $\hog(P)$ (and $Ov(P)$) from the Aho-Corasick trie of $P$ in $O(\norm{P})$ time. 

Let's assume we have the Aho-Corasick trie of $P$ including the failure links. We define $R(u)$ as follows:

\begin{equation}
R(u) = \{i \in \{1, \ldots, n\} \mid S(u) \text{ is a proper prefix of } s_i\}.
\end{equation}
That is, $R(u)$ is a set of string indices in the subtree rooted at $u$ if $u$ is an internal node, or an empty set if $u$ is a leaf node.

For each input string $s_i$, we will do the following operation separately, which is to find the longest overlap from $s_i$ to any string in $P$. Consider a path $(v_0, v_1, \ldots, v_l)$ which starts from the leaf representing $s_i$ and follows the failure links until it reaches the root, i.e., $S(v_0) = s_i$ and $v_l$ is the root of the tree. By definition of the failure link, the strings corresponding to nodes appearing in the path are the suffixes of $s_i$. If there are an index $j$ and a node $v_k$ on the path such that $j\in R(v_k)$, $S(v_k)$ is both a suffix of $s_i$ and a proper prefix of $s_j$, so $S(v_k)$ is an overlap from $s_i$ to $s_j$. 

$S(v_k)$ for $0 < k \leq l$ is the longest overlap from $s_i$ to $s_j$ if and only if $j \in R(v_k)$ and there is no $m$ such that $0 \le m < k$ and $j \in R(v_m)$. If there exists such $m$, then $S(v_m)$ is a longer overlap from $s_i$ to $s_j$ than $S(v_k)$, so $S(v_k)$ is not the longest overlap. Therefore, we get the following lemma.

\begin{lemma}
\label{lem:overlapcondition}
\emph{$S(v_k)$ is the longest overlap from $s_i$ to $s_j$ if and only if $j \in R(v_k) - R(v_{k-1}) - \ldots - R(v_0)$.}
\end{lemma}

Therefore, if $|R(v_k) - R(v_{k-1}) - \ldots - R(v_0)| > 0$, $S(v_k)$ is the longest overlap from $s_i$ to $s_j$ for $j \in R(v_k) - R(v_{k-1}) - \ldots - R(v_0)$, and thus $v_k \in Ov(P)$. Therefore, we aim to compute $|R(v_k) - R(v_{k-1}) - \ldots - R(v_0)|$ for every $0 < k \leq l$.

Given an index $k$, we define $k+1$ auxiliary sets of indices $I_k(k), I_k(k-1), \ldots, I_k(0)$ in a recursive manner as follows.

\begin{itemize}
\item $I_k(k) = R(v_k)$
\item $I_k(m) = I_k(m+1) - R(v_m)$ for $m = k-1, k-2, \ldots, 0$
\end{itemize}

By definition, $I_k(0)$ is $R(v_k) - R(v_{k-1}) - \ldots - R(v_0)$ in Lemma \ref{lem:overlapcondition} and we want to compute $|I_k(0)|$. For every $0 \leq m < k$, $I_k(m) = I_k(m+1) - R(v_m) \subseteq I_k(m+1)$ and thus $|I_k(m)| = |I_k(m+1)| - |I_k(m+1) - I_k(m)|$ holds. By summing up all these equations for $0 \leq m < k$, we get $|I_k(0)| = |I_k(k)| - \sum_{m=0}^{k-1} {|I_k(m+1) - I_k(m)|}$. Since $I_k(k) = R(v_k)$ and $I_k(m+1) - I_k(m) = I_k(m+1) - (I_k(m+1) - R(v_m)) = I_k(m+1) \cap R(v_m)$, we have

\begin{equation}
|I_k(0)| = |R(v_k)| - \sum_{m=0}^{k-1} {|I_k(m+1) \cap R(v_m)|}.
\label{eq:Ik}
\end{equation}

We also define a new function $up(u)$ for a node $u$ as follows.

\begin{definition}
\emph{Given a node $u$ in the Aho-Corasick trie, $up(u)$ is defined as the first ancestor of $u$ (except $u$ itself) in the path that starts at $u$ and follows the failure links until it reaches the root node. We define an ancestor on the tree which consists of tree edges in the Aho-Corasick trie.}
\end{definition}
\noindent Note that $up(u)$ is well defined when $u$ is not the root node, since the root node is always an ancestor of $u$. When $u$ is the root node, $up(u)$ is empty.

\begin{figure}[t]
    \centering
    \includegraphics[width=0.4\linewidth]{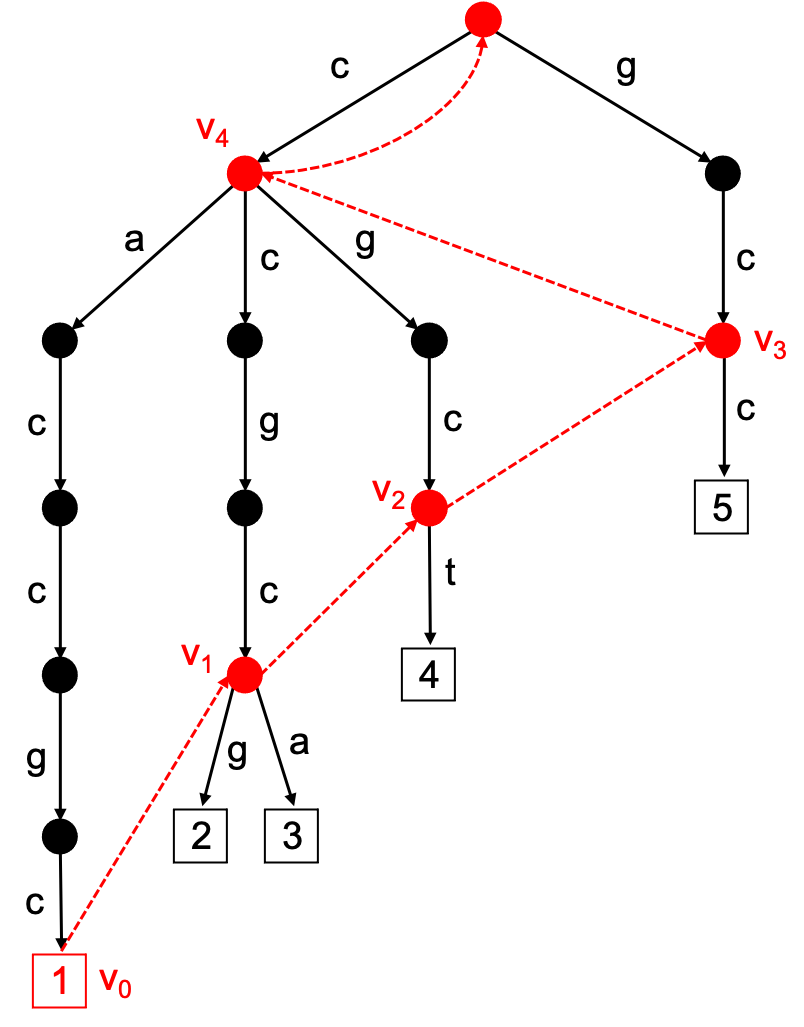}
    \caption{Aho-Corasick trie with $P = \{caccgc, ccgcg, ccgca, cgct, gcc\}$.}
    \label{fig:vpic}
\end{figure}

Now we analyze the value of $|I_k(m+1) \cap R(v_m)|$ in Equation (\ref{eq:Ik}) for each $0 \leq m < k$ as follows. We use a path $(v_0, v_1, ..., v_5)$ in Figure \ref{fig:vpic} as a running example, i.e., $l = 5$ and $0 < k \leq 5$.

\begin{lemma}
\label{lem:caseanalysis}
\emph{$|I_k(m+1) \cap R(v_m)|$ is $|R(v_m)|$ if $up(v_m) = v_k$; it is $0$ otherwise.}
\end{lemma}

\begin{proof}
We divide the relationship between $v_m$ and $v_k$ into cases.

\begin{enumerate}
\item $v_m$ is outside the subtree rooted at $v_k$

Let's assume that $I_k(m+1) \cap R(v_m)$ is not empty and there exists $j \in I_k(m+1) \cap R(v_m)$. Then $j \in R(v_k) \cap R(v_m)$ should hold since $I_k(m+1) \subseteq I_k(k) = R(v_k)$. Therefore, both $v_m$ and $v_k$ should be the ancestors of the leaf corresponding to $s_j$. Because $|S(v_m)| > |S(v_k)|$, $v_k$ should be an ancestor of $v_m$. Since $v_m$ is outside the subtree rooted at $v_k$, $v_k$ cannot be an ancestor of $v_m$, which is a contradiction. Therefore such $j$ does not exist, which shows that $I_k(m+1) \cap R(v_m) = \emptyset$ and $|I_k(m+1) \cap R(v_m)| = 0$. 

For example, consider the case with $k = 4$ and $m = 3$ in Figure \ref{fig:vpic}. Since $I_4(4) = R(v_4) = \{1, 2, 3, 4\}$ and $R(v_3) = \{5\}$, we can see that $I_4(4) \cap R(v_3) = \emptyset$.

\item $v_m$ is inside the subtree rooted at $v_k$

In this case, $v_k$ is an ancestor of $v_m$ and we further divide it into cases.
\begin{enumerate}
    \item There exists $q$ such that $m < q < k$ and $v_q$ is an ancestor of $v_m$.
    
    We get $R(v_m) \subseteq R(v_q)$ because $v_q$ is an ancestor of $v_m$. Since $I_k(m+1) = R(v_k) - R(v_{k-1}) - ... - R(v_{m+1})$ and $m < q < k$, we have $I_k(m+1) \subseteq R(v_k) - R(v_q)$. Therefore, $I_k(m+1) \cap R(v_m) \subseteq (R(v_k) - R(v_q)) \cap R(v_q) = \emptyset$. That is, $I_k(m+1) \cap R(v_m) = \emptyset$ and $|I_k(m+1) \cap R(v_m)| = 0$.
    
    \item For any $q$ such that $m < q < k$, $v_q$ is not an ancestor of $v_m$.
    
    Here we show that $R(v_m) \subseteq I_k(m+1)$. Let's consider an index $x \in R(v_m)$. Since $v_k$ is an ancestor of $v_m$, we have $x \in R(v_k)$. Moreover, for any $q$ such that $m < q < k$, neither $v_q$ is an ancestor of $v_m$ nor $v_m$ is an ancestor of $v_q$. That is, $R(v_q) \cap R(v_m) = \emptyset$ and thus $x \notin R(v_q)$. Therefore, we have $x \in I_k(m+1) = R(v_k)-R(v_{k-1})-\ldots-R(v_{m+1})$. In conclusion, $R(v_m) \subseteq I_k(m+1)$ and thus $|I_k(m+1) \cap R(v_m)| = |R(v_m)|$. 
    
    For example, consider the case with $k = 4$ and $m = 1$ in Figure \ref{fig:vpic}. Since $I_4(2) = R(v_4) - R(v_3) - R(v_2) = \{1, 2, 3\}$ and $R(v_1) = \{2, 3\}$, we can see that $R(v_1) \subseteq I_4(2)$ and $I_4(2) \cap R(v_1) = R(v_1)$.
\end{enumerate}

\end{enumerate}

In summary, $|I_k(m+1) \cap R(v_m)| = |R(v_m)|$ in case 2(b), and $0$ otherwise. In case 2(b), $v_k$ is an ancestor of $v_m$ and there is no $q$ such that $m < q < k$ and $v_q$ is an ancestor of $v_m$. In other words, $v_k$ is the first ancestor of $v_m$ in the path starting from $v_m$ and following the failure links repeatedly, which means that $up(v_m) = v_k$.

\end{proof}

\begin{theorem}
\label{thm:Ik}
\emph{For every $0 < k \leq l$, $|I_k(0)| = |R(v_k)| - \sum_{v_m} {|R(v_m)|}$, where $0\leq m<k$ and $up(v_m) = v_k$.}
\end{theorem}
\begin{proof}
From Equation (\ref{eq:Ik}), we have $|I_k(0)| = |R(v_k)| - \sum_{m=0}^{k-1} {|I_k(m+1) \cap R(v_m)|}$. By Lemma \ref{lem:caseanalysis}, we have $\sum_{m=0}^{k-1} {|I_k(m+1) \cap R(v_m)|} = \sum_{v_m : up(v_m) = v_k} {|R(v_m)|}$. By merging the two equations, we have the theorem.
\end{proof}

Now let's consider the relationship between $u$ and $up(u)$. $S(up(u))$ is a proper suffix of $S(u)$ because $up(u)$ can be reached from $u$ through failure links. Furthermore, $S(up(u))$ is a proper prefix of $S(u)$ because $up(u)$ is an ancestor of $u$. That is, $S(up(u))$ is a \textit{border} \cite{KMP} of $S(u)$. Moreover, we visit every suffix of $u$ in the decreasing order of lengths and $S(up(u))$ is the first border we visit, so $S(up(u))$ is the \textit{longest border} of $S(x)$. Since each node in the Aho-corasick trie corresponds to a prefix of some $s_i$, we can compute $up(u)$ for all nodes $u$ by computing the border array of every $s_i$ as follows. Let $pnode_i(l)$ be the node which corresponds to $s_i[1..l]$, and $border_i(l)$ be the length of the longest border of $s_i[1..l]$. Then we have the following equation for every $s_i$ and $1 \le l \le |s_i|$:
\begin{equation}
\label{eq:upborder}
up(pnode_i(l)) = pnode_i(border_i(l)).
\end{equation}
If we store $pnode_i$ and $border_i$ using arrays, we can compute $pnode_i, border_i$, and $up(u)$ in $O(\norm{P})$ time and space, because $border_i$ can be computed in $O(\norm{P})$ time using an algorithm in \cite{KMP}.

For example, let's consider Figure \ref{fig:vpic}, which is an Aho-Corasick trie built with a set $P = \{s_1 = caccgc, s_2 = ccgcg, s_3 = ccgca, s_4 = cgct, s_5 = gcc\}$ of strings. For each string, we compute its corresponding border array, and get $border_1 = (0, 0, 1, 1, 0, 1)$, $border_2 = (0, 1, 0, 1, 0)$, $border_3 = (0, 1, 0, 1, 0)$, $border_4 = (0, 0, 1, 0)$, and $border_5 = (0, 0, 0)$. We also store $pnode_i$'s by traversing the Aho-Corasick trie. Now we can compute $up$ by using $pnode_i$ and $border_i$. For example, let's consider $v_1 = pnode_2(4)$, which represents $ccgc$. Since the longest border of $ccgc$ is $c$, which has length $1$, we have $border_2(4) = 1$. As a result, we have $up(v_1) = up(pnode_2(4)) = pnode_2(border_2(4)) = pnode_2(1) = v_4$ by Equation (\ref{eq:upborder}). Note that $v_4$ represents $c$, which is the longest border of $ccgc$.

\begin{algorithm}[t]
\caption{Build HOG in linear time}\label{alg:computehog}
\begin{algorithmic}[1]
\Procedure{Build-HOG}{$P$}
    \State Build the Aho-Corasick trie with $P$
    \State Compute border arrays $border_i$ for $1 \leq i \leq n$
    \State Compute $up(u)$ for each node $u$
    \State Compute $|R(u)|$ for each node $u$
    \State Mark \texttt{root} as included in $\hog(P)$
    \State For each node $u$, initialize $\text{Child}(u)$ with an empty set
    \For{$i \gets 1$ \textbf{to} $n$}
        \State $u \gets$ leaf corresponding to $s_i$ in Aho-Corasick trie
        \State Mark $u$ as included in $\hog(P)$
        \While{$u \neq \mathtt{root}$}
            \State $I(u) \gets |R(u)|$
            \ForAll{$u' \in \text{Child}(u)$}
                \State $I(u) \gets I(u) - |R(u')|$
            \EndFor
            \If{$I(u) > 0$}
                \State Mark $u$ as included in HOG(P)
            \EndIf
            \State $\text{Child}(u) \gets$ an empty set
            \State Add $u$ to $\text{Child}(up(u))$
            \State $u \gets$ failure link of $u$
        \EndWhile
    \EndFor
    \State Build HOG(P) with marked nodes
\EndProcedure
\end{algorithmic}
\end{algorithm}

We are ready to describe an algorithm to compute HOG of $P$ in $O(\norm{P})$ time and space. First, we build the Aho-Corasick trie with $P$ and a border array for each $s_i$. By using the border arrays, we compute $up(u)$ for every node $u$ except the root. Next, we compute $|R(u)|$ for each node $u$ by the post-order traversal of the Aho-Corasick trie. For each string $s_i$, we start from the leaf node corresponding to $s_i$ and follow the failure links until we reach the root. For each node $v_k$ that we visit, we compute its corresponding $|I_k(0)| = |R(v_k)| - \sum_{v_m : up(v_m) = v_k} {|R(v_m)|}$. If $|I_k(0)| > 0$, we mark $v_k$ to be included in HOG. Algorithm \ref{alg:computehog} shows an algorithm to compute HOG. Lines 2-5 compute the preliminaries for the algorithm, while lines 6-19 compute the nodes to be included in HOG. Note that the loop of lines 8-19 works separately for each input string $s_i$. We consider $v_k$ in the order of increasing $k$, and thus if $up(v_m) = v_k$, then $m < k$. Hence, $\text{Child}(v_k)$ in line 13 stores every $v_m$ such that $up(v_m) = v_k$ by line 18 of previous iterations. For each node $u = v_k$ in lines 11-19, $I(u)$ correctly computes $|I_k(0)|$ since we get $|R(v_k)|$ in line 12 and subtract every $|R(v_m)|$ where $v_k = up(v_m)$ in lines 13-14. According to Theorem \ref{thm:Ik}, lines 12-14 correctly computes $|I_k(0)|$. We build $\hog(P)$ in line 20 by removing the unmarked nodes and contracting the edges while traversing the Aho-Corasick trie once, as in \cite{hogipl}.

For example, let's consider a set $P = \{s_1 = caccgc, s_2 = ccgcg, s_3 = ccgca, s_4 = cgct, s_5 = gcc\}$ of strings. Figure \ref{fig:vpic} shows the Aho-Corasick trie built with $P$. Consider a path starting from a node representing $s_1$ and following the failure links until the root node. The path $(v_0, v_1, v_2, v_3, v_4, v_5)$ is marked with dotted lines in Figure \ref{fig:vpic}. By definition of $up$, we get $up(v_0) = up(v_1) = up(v_2) = v_4$ and $up(v_3) = up(v_4) = v_5$. Therefore, we can compute $|I_k(0)|$'s as follows.

\begin{itemize}
    \item $|I_0(0)| = |R(v_0)| = 0$
    \item $|I_1(0)| = |R(v_1)| = 2$
    \item $|I_2(0)| = |R(v_2)| = 1$
    \item $|I_3(0)| = |R(v_3)| = 1$
    \item $|I_4(0)| = |R(v_4)| - |R(v_0)| - |R(v_1)| - |R(v_2)| = 4 - 0 - 2 - 1 = 1$
\end{itemize}

Note that $|R(v_0)| = 0$ by definition of $R(u)$. Since $v_1, v_2, v_3,$ and $v_4$ have positive $|I_k(0)|$'s, we mark them to be included in HOG. We do this process for every $s_i$.

Now we show that Algorithm \ref{alg:computehog} runs in $O(\norm{P})$ time and space. Computing an Aho-Corasick trie, a border array for each string, and $up(u)$ and $|R(u)|$ for each node $u$ costs $O(\norm{P})$ time and space. Furthermore, for a given index $i$, lines 13-14 are executed at most $|s_i|$ times since line 18 is executed at most $|s_i|$ times, and thus the sum of $|\text{Child}(u)|$ is at most $|s_i|$. Therefore, lines 9-19 run in $O(|s_i|)$ time for given $i$, and thus lines 8-19 run in $O(\norm{P})$ time in total. Also they use $O(|s_i|)$ additional space to store the $\text{Child}$ list. Lastly, we can build $\hog(P)$ with marked nodes in $O(\norm{P})$ space and time \cite{hogipl}. Therefore, Algorithm \ref{alg:computehog} runs in $O(\norm{P})$ time and space. We remark that Algorithm \ref{alg:computehog} can be modified so that it builds the HOG from an EHOG instead of an Aho-Corasick trie, while it still costs $O(\norm{P})$ time and space.

\begin{theorem}
\emph{Given a set $P$ of strings, $\hog(P)$ can be built in $O(\norm{P})$ time and space.}
\end{theorem}


\section{Conclusion}
We have presented an $O(\norm{P})$ time and space algorithm to build the HOG, which improves upon an earlier $O(\norm{P} \log n)$ time solution. Since the input size of the problem is $O(\norm{P})$, the algorithm is optimal.

There are some interesting topics about HOG and EHOG which deserve the future work. As mentioned in the introduction, the \emph{shortest superstring problem} gained a lot of interest \cite{Blum94,Sweedyk00,Mucha13,Paluch14}. Since many algorithms dealing with the shortest superstring problem are based on the overlap graph, HOG may give better approximation algorithms for the shortest superstring problem by using the additional information that HOG has when compared to the overlap graph. 


\section*{Acknowledgements}
S. Park, S.G. Park and K. Park were supported by Institute for Information \& communications Technology Promotion(IITP) grant funded by the Korea government(MSIT) (No. 2018-0-00551, Framework of Practical Algorithms for NP-hard Graph Problems). B. Cazaux and E. Rivals acknowledge the funding from Labex NumeV (GEM flagship project, ANR 2011-LABX-076), and from the Marie Skłodowska-Curie Innovative Training Networks  ALPACA (grant 956229).

\bibliography{references}

\end{document}